\documentclass[english]{ics}
\pdfoutput=1
\usepackage{units}
\usepackage{amsfonts}
\usepackage{bbm}
\usepackage{fancyhdr}

\usepackage[]{amsmath,amssymb,amsfonts,latexsym,amsthm,enumerate,url}
\usepackage{times}

\usepackage[colorlinks=true]{hyperref}
\begin{document}
\renewcommand{\citet}{\cite}

\rhead{MIT-CTP-4091}

\title{Breaking and making quantum money: toward a new quantum cryptographic
protocol}

\author{Andrew~Lutomirski$^1$ \and Scott~Aaronson$^{2}$ \and Edward~Farhi$^{1}$ \and David~Gosset$^{1}$ \and Avinatan~Hassidim${^1}$ \and Jonathan~Kelner$^{2,3}$ \and Peter~Shor$^{1,2,3}$}

\address{
$^{1}$Center for Theoretical Physics, Massachusetts Institute of Technology, Cambridge, MA 02139
\and
$^{2}$Computer Science and Artificial Intelligence Laboratory, Massachusetts Institute of Technology, Cambridge, MA 02139
\and
$^{3}$Department of Mathematics, Massachusetts Institute of Technology, Cambridge, MA 02139}
\email{luto@mit.edu \and aaronson@csail.mit.edu \and farhi@mit.edu \and dgosset@mit.edu \and avinatanh@gmail.com \and kelner@mit.edu \and shor@math.mit.edu} 

\begin{abstract}
Public-key quantum money is a cryptographic protocol in which a bank
can create quantum states which anyone can verify but no one except
possibly the bank can clone or forge. There are no secure public-key
quantum money schemes in the literature; as we show in this paper,
the only previously published scheme \citet{aaronson-quantum-money}
is insecure. We introduce a category of quantum money protocols which
we call \emph{collision-free}. For these protocols, even the bank
cannot prepare multiple identical-looking pieces of quantum money.
We present a blueprint for how such a protocol might work as well
as a concrete example which we believe may be insecure. 
\end{abstract}

\maketitle
\renewcommand{\headrulewidth}{0pt}
\thispagestyle{fancyplain}

\global\long\def\Tr{\mathop\mathrm{Tr}}

\global\long\def\poly{\mathop\mathrm{poly}}

\global\long\def\E{\mathop\mathbb{E}}

\global\long\def\ket#1{|#1\rangle}

\global\long\def\bra#1{\langle#1|}

\newtheorem{definition}{Definition}
\newtheorem{lemma}[definition]{Lemma}
\newtheorem{thm}[definition]{Theorem}

\section{Introduction}

In 1969, Wiesner \citet{wiesner} pointed out that the no-cloning
theorem raises the possibility of uncopyable cash: bills whose authenticity
would be guaranteed by quantum physics.%
\footnote{This is the same paper that introduced the idea of quantum cryptography.
Wiesner's paper was not published until the 1980s; the field of quantum
computing and information (to which it naturally belonged) had not
yet been invented.%
} Here's how Wiesner's scheme works: besides an ordinary serial number,
each bill would contain (say) a few hundred photons, which the central
bank polarized in random directions when it issued the note. The bank
remembers the polarization of every photon on every bill ever issued.
If you want to verify that a bill is genuine, you take it to the bank,
and the bank uses its knowledge of the polarizations to measure the
photons. On the other hand, the No-Cloning Theorem ensures that someone
who \textit{doesn't} know the polarization of a photon can't produce
more photons with the same polarizations. Indeed, copying a bill can
succeed with probability at most $\left(5/6\right)^{n}$, where $n$
is the number of photons per bill.

Despite its elegance, Wiesner's quantum money is a long way from replacing
classical money. The main practical problem is that we don't know
how to reliably store polarized photons (or any other coherent quantum
state) for any appreciable length of time.

Yet, even if we could solve the technological problems, Wiesner's
scheme would still have a serious drawback: only the bank can verify
that a bill is genuine. Ideally, \textit{printing} bills ought to
be the exclusive prerogative of the bank, but the \textit{checking}
process ought to be open to anyone---think of a convenience-store
clerk holding up a \$20 bill to a light.

But, with quantum mechanics, it may be possible to have quantum money
satisfying all three requirements:
\begin{enumerate}
\item The bank can print it. That is, there is an efficient algorithm to
produce the quantum money state.
\item Anyone can verify it. That is, there is an efficient measurement that
anyone can perform that accepts money produced by the bank with high
probability and minimal damage.
\item No one (except possibly the bank) can copy it. That is, no one other
than the bank can efficiently produce states that are accepted by
the verifier with better than exponentially small probability.
\end{enumerate}
We call such a scheme a \textit{public-key quantum money scheme},
by analogy with public-key cryptography. Such a scheme cannot be secure
against an adversary with unbounded computational power, since a brute-force
search will find valid money states in exponential time. Surprisingly,
the question of whether public-key quantum money schemes are possible
under computational assumptions has remained open for forty years,
from Wiesner's time until today.

The first proposal for a public-key quantum money scheme, along with
a proof that such money exists in an oracle model, appeared in \citet{aaronson-quantum-money}.
We show in section~\ref{sec:stabilizer-money-insecure} that the
proposed quantum money scheme is insecure.

In this paper we introduce the idea of \emph{collision-free} quantum
money, which is public-key quantum money with the added restriction
that no one, not even the bank, can efficiently produce two identical-looking
pieces of quantum money. We discuss the prospect of implementing collision-free
quantum money and its uses in section~\ref{sec:kinds-of-money} below.

The question of whether secure public-key quantum money exists remains
open.

\section{Two kinds of quantum money\label{sec:kinds-of-money}}

All public-key quantum money schemes need some mechanism to identify
the bank and prevent other parties from producing money the same way
that the bank does. A straightforward way of accomplishing this is
to have the money consist of a quantum state and a classical description,
digitally signed by the bank, of a circuit to verify the quantum state.
Digital signatures secure against quantum adversaries are believed
to exist, so we do not discuss the signature algorithm in the remainder
of the paper.

Alternatively, if the bank produces a fixed number of quantum money
states, it could publish a list of all the verification circuits of
all the valid money states, and anyone could check that the verifier
of their money state is in that list. This alternative is discussed
further in section~\ref{sub:collision-free-money}.

\subsection{Quantum money with a classical secret}

Public-key quantum money is a state which can be produced by a bank
and verified by anyone. One way to design quantum money is to have
the bank choose, for each instance of the money, a classical secret
which is a description of a quantum state that can be efficiently
generated and use that secret to manufacture the state. The bank then
constructs an algorithm to verify that state and distributes the state
and a description of the algorithm as {}``quantum money.'' We will
refer to protocols of this type as \emph{quantum money with a classical
secret}. The security of such a scheme relies on the difficulty of
deducing the classical secret given both the verification circuit
and a copy of the state.

A simple but insecure scheme for this type of quantum money is based
on random product states. The bank chooses a string of $n$ uniformly
random angles $\theta_{i}$ between $0$ and $2\pi$. This string
is the classical secret. Using these angles, the bank generates the
state $|\psi\rangle=\otimes_{i}|\theta_{i}\rangle$ where $|\theta_{i}\rangle=\cos\theta_{i}|0\rangle+\sin\theta_{i}|1\rangle$
and chooses a set of (say) 4-local projectors which are all orthogonal
to $|\psi\rangle$. The quantum money is the state $|\psi\rangle$
and a classical description of the projectors, and anyone can verify
the money by measuring the projectors.

It is NP-hard to produce the state $|\psi\rangle$ given only a description
of the projectors, and given only the state, the no-cloning theorem
states that the state cannot be copied. However, this quantum money
is insecure because of a fully quantum attack \citet{state-restoration}
that uses a copy of the state and the description of the projectors
to produce additional copies of the state. A more sophisticated example
of quantum money with a classical secret is described in \citet{aaronson-quantum-money}.
A different scheme was proposed Mosca and Stebila in \citet{mosca-2009}.
The latter scheme requires a classical oracle that we do not know
how to construct.

All quantum money schemes which rely on a classical secret in this
way have the property, shared with ordinary bank notes and coins,
that an unscrupulous bank can produce multiple pieces of identical
money. Also, if there is a classical secret, there is the risk that
some classical algorithm can deduce the secret from the verification
algorithm (we show in section~\ref{sec:stabilizer-money-insecure}
that the scheme of \citet{aaronson-quantum-money} fails under some
circumstances for exactly this reason).

\subsection{Collision-free quantum money\label{sub:collision-free-money}}

An alternative kind of quantum money is \emph{collision-free}. This
means that the bank cannot efficiently produce two pieces of quantum
money with the same classical description of the verification circuit.
This rules out protocols in which the verification circuit is associated
with a classical secret which allows the bank to produce the state.
(For example, in the product state construction in the previous section,
the set of angles would allow the bank to produce any number of identical
pieces of quantum money.)

Collision-free quantum money has a useful property that even uncounterfeitable
paper money (if it existed) would not have: instead of just digitally
signing the verification circuit for each piece of money, the bank
could publish a list describing the verification circuit of each piece
of money it intends to produce. These verification circuits would
be like serial numbers on paper money, but, since the bank cannot
cheat by producing two pieces of money with the same serial number,
it cannot produce more money than it says. This means that the bank
cannot inflate the currency by secretly printing extra money.

We expect that computationally secure collision-free quantum money
is possible. We do not have a concrete implementation of such a scheme,
but in the next few sections, we give a blueprint for how a collision-free
quantum money scheme could be constructed. We hope that somebody produces
such a scheme which will not be vulnerable to attack.

\subsubsection{Quantum money by postselection\label{sub:postselection-money}}

Our approach to collision-free quantum money starts with a classical
set. For concreteness, we will take this to be the set of $n$-bit
strings. We need a classical function $L$ that assigns a label to
each element of the set. There should be an exponentially large set
of labels and an exponentially large number of elements with each
label. Furthermore, no label should correspond to more than an exponentially
small fraction of the set. The function $L$ should be as obscure
and have as little structure as possible. The same function can be
used to generate multiple pieces of quantum money. Each piece of quantum
money is a state of the form \[
|\psi_{\ell}\rangle=\frac{1}{\sqrt{N_{\ell}}}\sum_{x\mbox{ s.t. }L\left(x\right)=\ell}|x\rangle\]
 along with the label $\ell$ which is used as part of the verification
procedure ($N_{\ell}$ is the number of terms in the sum). The function
$L$ must have some additional structure in order to verify the state.

Such a state can be generated as follows. First, produce the equal
superposition over all $n$-bit strings. Then compute the function
$L$ into an ancilla register and measure that register to obtain
a particular value $\ell$. The state left over after measurement
will be $|\psi_{\ell}\rangle$.

The quantum money state $|\psi_{\ell}\rangle$ is the equal superposition
of exponentially many terms which seemingly have no particular relationship
to each other. Since no label occurs during the postselection procedure
above with greater than exponentially small probability, the postselection
procedure would have to be repeated exponentially many times to produce
the same label $\ell$ twice. If the labeling function $L$ is a black
box with no additional structure, then Grover's lower bound rules
out any polynomial time algorithm that can produce the state $\ket{\psi_{\ell}}$
given only knowledge of $\ell$. We conjecture that it is similarly
difficult to copy a state $|\psi_{\ell}\rangle$ or to produce the
state $\ket{\psi_{\ell}}\otimes|\psi_{\ell}\rangle$ for any $\ell$
at all. 

It remains to devise an algorithm to verify the money.

\subsubsection{Verification using rapidly mixing Markov chains\label{sub:verification-with-markov-chains}}

The first step of any verification algorithm is to measure the function
$L$ to ensure that the state is a superposition of basis vectors
associated with the correct label $\ell$. The more difficult task
is to verify that it is the correct superposition $|\psi_{\ell}\rangle.$ 

Our verification procedure requires some additional structure in the
function $L$: we assume that we know of a classical Markov matrix
$M$ which, starting from any distribution over bit strings with the
same label $\ell$, rapidly mixes to the uniform distribution over
those strings but does not mix between strings with different $\ell$.
This Markov chain must have a special form: each update must consist
of a uniform random choice over $N$ update rules, where each update
rule is deterministic and invertible. We can consider the action of
the operator $M$ on the Hilbert space in which our quantum money
lives ($M$ is, in general, neither unitary nor Hermitian). Acting
on states in this Hilbert space, any valid quantum money state $\ket{\psi_{\ell}}$
is a +1 eigenstate of $M$ and, in fact, \begin{equation}
M^{r}\approx\sum_{l}|\psi_{\ell}\rangle\langle\psi_{\ell}|\label{eq:M_to_the_r}\end{equation}
 where the approximation is exponentially good for polynomially large
$r$. This operator, when restricted to states with a given label
$\ell$, approximately projects onto the money state $|\psi_{\ell}\rangle$.
After measuring the label $\ell$ as above, the final step of our
verification procedure is to measure $M^{r}$ for sufficiently large
$r$ as we describe below. Even using the Markov chain $M$, we do
not know of a general way to efficiently copy quantum money states
$\ket{\psi_{\ell}}$. 

Any deterministic, invertible function corresponds to a permutation
of its domain; we can write the Markov matrix as the average of $N$
such permutations $P_{i}$ over the state space, where $P_{i}$ corresponds
to the $i^{\text{th}}$ update rule. That is \[
M=\frac{1}{N}\sum_{i=1}^{N}P_{i}.\]

We define a controlled update $U$ of the state, which is a unitary
quantum operator on two registers (the first holds an $n$-bit string
and the second holds numbers from 1 to $N$) \[
U=\sum_{i}P_{i}\otimes\ket i\bra i.\]

Given some initial quantum state on $n$ qubits, we can add an ancilla
in a uniform superposition over all $i$ (from 1 to $N$). We then
apply the unitary $U$, measure the projector of the ancilla onto
the uniform superposition, and discard the ancilla. The Kraus operator
sum element corresponding to the outcome 1 is \begin{align*}
 & \quad\left(I\otimes\frac{1}{\sqrt{N}}\sum_{i=1}^{N}\bra i\right)U\left(I\otimes\frac{1}{\sqrt{N}}\sum_{i=1}^{N}\ket i\right)\\
 & =\frac{1}{N}\sum_{i=1}^{N}P_{i}\\
 & =M.\end{align*}
This operation can be implemented with one call to controlled-$P_{i}$
and additional overhead logarithmic in $N$. Repeating this operation
$r$ times, the Kraus operator corresponding to all outcomes being
1 is $M^{r}$. The probability that all of the outcomes are 1 starting
from a state $\ket{\phi}$ is $\left\Vert M^{r}|\phi\rangle\right\Vert ^{2}$
and the resulting state is $M^{r}|\phi\rangle/\left\Vert M^{r}|\phi\rangle\right\Vert ^{2}$.
If choose a large enough number of iterations $r$, we approximate
a measurement of $\sum_{l}|\psi_{\ell}\rangle\langle\psi_{\ell}|$
as in eq.~\ref{eq:M_to_the_r}.

This construction has the caveat that, if the outcomes are not all
1, the final state is not $(1-M^{r})\ket{\psi}$. This can be corrected
by deferring all measurements, computing an indicator of whether all
outcomes were 1, and uncomputing everything else, but, as we do not
care about the final state of bad quantum money, we do not need this
correction.

\subsection{An example of quantum money by postselection}

\subsubsection{Constructing a label function}

One approach to creating the labeling function $L$ from Sec.~\ref{sub:postselection-money}
is to concatenate the output of multiple single-bit classical cryptographic
hash functions,%
\footnote{A simpler apprach would be to hash the entire $n$-bit string onto
a smaller, but still exponentially large, set of labels. We do not
pursue this approach because we do not know of any way to verify the
resulting quantum money states.%
} each of which acts on some subset of the qubits in the money state.
We will describe such a scheme in this section, which has promising
properties but is most likely insecure.

We start by randomly choosing $\left\lceil \sqrt{n}\right\rceil $
subsets of the $n$ bits, where each bit is in 10 of the subsets.
We associate a different binary valued hash function with each subset.
The hash function associated with a particular subset maps the bits
in that subset to either 0 or 1. The labeling function $L$ is the
$\left\lceil \sqrt{n}\right\rceil $-bit string which contains the
outputs of all the hash functions.

The bank can produce a random pair $\left(\ell,|\psi_{\ell}\rangle\right)$,
where $|\psi_{\ell}\rangle$ is the uniform superposition of all bit
strings that hash to the values corresponding to the label $\ell$,
by using the algorithm in Sec.~\ref{sub:postselection-money}.

\subsubsection{Verifying the Quantum Money}

As in Sec.~\ref{sub:verification-with-markov-chains}, we verify
the money using a Markov chain. The update rule for the Markov chain
is to choose a bit at random and flip the bit if and only if flipping
that bit would not change the label (i.e. if all of the hash function
that include that bit do not change value, which happens with roughly
constant probability). This Markov chain is not ergodic, because there
are probably many assignments to all the bits which do not allow any
of the bits to be flipped. These assignments, along with some other
possible assignments that mix slowly, can be excluded from the superposition,
and the verification circuit may still be very close to a projector
onto the resulting money state.

\subsubsection{A weakness of this quantum money}

A possible weakness of our hash-based labeling function as defined
above is that the label is not an opaque value---the labels of two
different bit strings are related to the difference between those
strings. Specifically, the problem of finding strings that map to
a particular label $\ell$ is a constraint satisfaction problem, and
the Hamming distance between the label $\ell'=L\left(x\right)$ and
$\ell$ is the number of clauses that the string $x$ violates.

We are concerned about the security of this scheme because it may
be possible to use the structure of the labeling function to implement
algorithms such as the state generation algorithm in \citet{adiabatic-state-generation},
which, under certain circumstances, could be used to produce the money
state. For example, consider a thermal distribution for which each
bit string has probability proportional to $e^{-\beta c\left(x\right)}$,
where $\beta$ is an arbitrary constant and $c\left(x\right)$ is
the number of clauses that the string $x$ violates. If for all $\beta$
we could construct a rapidly mixing Markov chain with this stationary
distribution, then we could apply the state generation algorithm mentioned
above. A naive Metropolis-Hastings construction that flips single
bits gives Markov chains that are not rapidly mixing at high $\beta$,
but some variants may be rapidly mixing. We do not know whether quantum
sampling algorithms based on such Markov chains can run in polynomial
time.

Due to this type of attack, and because we do not have a security
proof, we do not claim that this money is secure.

\section{Insecurity of a previously published quantum money scheme\label{sec:stabilizer-money-insecure}}

The only currently published public-key quantum money scheme, an example
of quantum money with a classical secret, was proposed in \citet{aaronson-quantum-money}.
We refer to this scheme as stabilizer money. We show that stabilizer
money is insecure by presenting two different attacks that work in
different parameter regimes. For some parameters, a classical algorithm
can recover the secret from the description of the verification circuit.
For other parameters, a quantum algorithm can generate states which
are different from the intended money state but which still pass verification
with high probability. Neither attack requires access to the original
money state.

The stabilizer money is parametrized by integers $n,$ $m$ and $l$
and by a real number $\epsilon\in[0,1]$. These parameters are required
to satisfy $\frac{1}{\epsilon^{2}}\ll l$.

The quantum money state is a tensor product of $l$ different stabilizer
states, each on $n$ qubits, and the classical secret is a list of
Pauli group operators which stabilize the state. The bank generates
an instance of the money by choosing a random stabilizer state for
each of the $l$ registers. To produce the verification circuit, the
bank generates an $m\times l$ table of $n$ qubit Pauli group operators.
The $(i,j)$th element of the table is an operator \[
E_{ij}=(-1)^{b_{ij}}A_{1}^{ij}\otimes A_{2}^{ij}...\otimes A_{n}^{ij}\]
 where each $A_{k}^{ij}\in\{1,\sigma_{x},\sigma_{y},\sigma_{z}\}$
and $b_{ij}\in\{0,1\}$. Each element $E_{ij}$ of the table is generated
by the following procedure: 
\begin{enumerate}
\item With probability $1-\epsilon$ choose the $b_{ij}$ and, for each
$k$, $A_{k}^{ij}$ uniformly at random.
\item With probability $\epsilon$ choose the operator $E_{ij}$ to be a
uniformly random element of the stabilizer group of $|C_{i}\rangle$.
\end{enumerate}
To verify the quantum money state, for each $i$ the authenticator
chooses $j\left(i\right)\in\left[m\right]$ at random and measures
\begin{equation}
Q=\frac{1}{l}\sum_{i}I^{\otimes i-1}\otimes E_{i,j\left(i\right)}\otimes I^{\otimes m-i}.\label{eq:operatorQ}\end{equation}
 The authenticator accepts iff the outcome is greater than or equal
to $\frac{\epsilon}{2}$. Note that measuring the operator $Q$ is
equivalent to measuring the operator $E_{i,j(i)}$ for each register
$i\in[l]$ and then averaging the results, since the measurements
on different registers commute.

The state $|C_{1}\rangle|C_{2}\rangle...|C_{l}\rangle$ is accepted
by this procedure with high probability since the probability of measuring
a $+1$ for the operator $E_{i,j(i)}$ on the state $|C_{i}\rangle$
is $\frac{1+\epsilon}{2}$ . The mean value of the operator $Q$ in
the state $|C_{1}\rangle|C_{2}\rangle...|C_{l}\rangle$ is therefore
$\epsilon$, since it is simply the average of the $E_{i,j(i)}$ for
each register $i\in[l]$. The parameter $l$ is chosen so that $\frac{l}{\epsilon^{2}}=\Omega\left(n\right)$
so the probability that one measures $Q$ to be less than $\frac{\epsilon}{2}$
is exponentially small in $n$.

Our attack on this money depends on the parameter $\epsilon$. Our
proofs assume that $m=\poly(n)$, but we expect that both attacks
work beyond the range in which our proofs apply.

\subsection{Attacking the verification circuit for $\epsilon\leq\frac{1}{16\sqrt{m}}$ }

For $\epsilon\leq\frac{1}{16\sqrt{m}}$ and with high probability
over the table of Pauli operators, we can efficiently generate a state
that passes verification with high probability. This is because the
verification algorithm does not project onto the intended money state
but in fact accepts many states with varying probabilities. On each
register, we want to produce a state for which the expected value
of the measurement of a random operator from the appropriate column
of $E$ is sufficiently positive. This is to ensure that, with high
probability, the verifier's measurement of $Q$ will have an outcome
greater than $\frac{\epsilon}{2}$. For small $\epsilon,$ there are
many such states on each register and we can find enough of them by
brute force.

We find states that pass verification by working on one register at
a time. For each register $i$, we search for a state $\rho_{i}$
with the property that \begin{equation}
\Tr\left[\left(\frac{1}{m}\sum_{j=1}^{m}E_{ij}\right)\rho_{i}\right]\geq\frac{1}{4\sqrt{m}}+O\left(\frac{1}{m^{2}}\right).\label{eq:bruteforce-low-epsilon-cond}\end{equation}
 As we show in Appendix~\ref{sec:details-of-low-epsilon-attack},
we can find such states efficiently on enough of the registers to
construct a state that passes verification.

\subsection{Recovering the classical secret for $\epsilon\geq\frac{c}{\sqrt{m}}$}

We describe how to recover the classical secret (i.e. a description
of the quantum state), and thus forge the money, when the parameter
$\epsilon\geq\frac{c}{\sqrt{m}}$ for any constant $c>0$. We observe
that each column of the table $E$ contains approximately $\epsilon m$
commuting operators, with the rest chosen randomly, and if, in each
column, we can find a set of commuting operators that is at least
as large as the planted set, then any quantum state stabilized by
these operators will pass verification.

We begin by casting our question as a graph problem. For each column,
let $G$ be a graph whose vertices correspond to the $m$ measurements,
and connect vertices $i$ and $j$ if and only if the corresponding
measurements commute. The vertices corresponding to the planted commuting
measurements now form a clique, and we aim to find it.

In general, it is intractable to find the largest clique in a graph.
In fact, it is NP-hard even to approximate the size of the largest
clique within $n^{1-\epsilon}$, for any $\epsilon>0$~\citet{Zuckerman}.
Finding large cliques planted in otherwise random graphs, however,
can be easy.

For example, if $\epsilon=\Omega\left(\frac{\log m}{\sqrt{m}}\right)$,
then a simple classical algorithm will find the clique. This algorithm
proceeds by sorting the vertices in decreasing order of degree and
selecting vertices from the beginning of the list as long as the selected
vertices continue to form a clique.

We can find the planted clique for $\epsilon\geq\frac{c}{\sqrt{m}}$
for any constant $c>0$ in polynomial time using a more sophisticated
classical algorithm that may be of independent interest. If the graph
were obtained by planting a clique of size $\epsilon\sqrt{m}$ in
a random graph drawn from $G(m,1/2)$, Alon, Krivelevich, and Sudakov
showed in \citet{AKS98} that one can find the clique in polynomial
time with high probability.%
\footnote{Remember that $G\left(m,p\right)$ is the Erdös-Rényi distribution
over $m$-vertex graphs in which an edge connects each pair of vertices
independently with probability $p$. The AKS algorithm was later improved
\citep{FK00} to work on subgraphs of $G(n,p)$ for any constant $p$,
but our measurement graph $G$ is not of that form. %
} Unfortunately, the measurement graph $G$ is not drawn from $G(m,1/2)$,
so we cannot directly apply their result. However, we show in appendix\textbf{~\ref{sec:details-of-low-epsilon-attack}}
that if $G$ is sufficiently random then a modified version of their
algorithm works.

\section{Conclusions}

Quantum money is an exciting and open area of research. Wiesner's
original scheme is information-theoretically secure, but is not public-key.
In this paper, we proved that the stabilizer construction for public-key
quantum money \citet{aaronson-quantum-money} is insecure for most
choices of parameters, and we expect that it is insecure for all choices
of parameters. We drew a distinction between schemes which use a classical
secret and those which are collision-free. We gave a blueprint for
how a collision-free scheme might be devised. We described an illustrative
example of such a scheme, but we have serious doubts as to its security.

It remains a major challenge to base the security of a public-key
quantum money scheme on any previously-studied (or at least standard-looking)
cryptographic assumption, for example, that some public-key cryptosystem
is secure against quantum attack. Much as we wish it were otherwise,
it seems possible that public-key quantum money intrinsically requires
a new mathematical leap of faith, just as public-key cryptography
required a new leap of faith when it was first introduced in the 1970s.

\section{Acknowledgments}

This work was supported in part by funds provided by the U.S. Department
of Energy under cooperative research agreement DE-FG02-94ER40818,
the W. M. Keck Foundation Center for Extreme Quantum Information Theory,
the U.S. Army Research Laboratory's Army Research Office through grant
number W911NF-09-1-0438, the National Science Foundation through grant
numbers CCF-0829421, CCF-0843915, and CCF-0844626, a DARPA YFA grant,
the NDSEG fellowship, the Natural Sciences and Engineering Research
Council of Canada, and Microsoft Research.

\bibliographystyle{hplain}
\addcontentsline{toc}{section}{\refname}\bibliography{money-bib}

\appendix

\section{Details of the attack against stabilizer money for $\epsilon\leq\frac{1}{16\sqrt{m}}$
\label{sec:details-of-low-epsilon-attack}}

For $\epsilon\leq\frac{1}{16\sqrt{m}}$ and with high probability
in the table of Pauli operators, we can efficiently generate a state
that passes verification with high probability. Our attack may fail
for some choices of the table used in verification, but the probability
that such a table of operators is selected by the bank is exponentially
small.

Recall that each instance of stabilizer money is verified using a
classical certificate, which consists of an $m\times l$ table of
$n$ qubit Pauli group operators. The $(i,j)$th element of the table
is an operator \[
E_{ij}=(-1)^{b_{ij}}A_{1}^{ij}\otimes A_{2}^{ij}...\otimes A_{n}^{ij}\]
 where each $A_{k}^{ij}\in\{1,\sigma_{x},\sigma_{y},\sigma_{z}\}$
and $b_{ij}\in\{0,1\}$.

We will use one important property of the algorithm that generates
the table of Pauli operators: with the exception of the fact that
$-I^{\otimes n}$ cannot occur in the table, the distribution of the
tables is symmetric under negation of all of the operators.

The verification algorithm works by choosing, for each $i$, a random
$j\left(i\right)\in\left[m\right]$. The verifier then measures \begin{equation}
Q=\frac{1}{l}\sum_{i}I^{\otimes i-1}\otimes E_{i,j\left(i\right)}\otimes I^{\otimes m-i}.\label{eq:operatorQ_appendix}\end{equation}
 The algorithm accepts iff the outcome is greater than or equal to
$\frac{\epsilon}{2}$. Note that measuring the operator $Q$ is equivalent
to measuring the operator $E_{i,j(i)}$ for each register $i\in[l]$
and then averaging the results, since the measurements on different
registers commute.

To better understand the statistics of the operator $Q$, we consider
measuring an operator $E_{i,j(i)}$ on a state $\rho_{i}$, where
$j(i)\in[m]$ is chosen uniformly at random. The total probability
$p_{1}(\rho_{i})$ of obtaining the outcome $+1$ is given by\begin{align*}
p_{1}(\rho_{i}) & =\frac{1}{m}\sum_{j=1}^{m}\Tr\left[\left(\frac{1+E_{i,j(i)}}{2}\right)\rho_{i}\right]\\
 & =\frac{1+\Tr\left[H^{(i)}\rho_{i}\right]}{2}\end{align*}
 where (for each $i\in[l]$) we have defined the Hamiltonian \[
H^{(i)}=\frac{1}{m}\sum_{j=1}^{m}E_{ij}.\]

We use the algorithm described below to independently generate an
$n$ qubit mixed state $\rho_{i}$ on each register $i\in[l]$. At
least $\nicefrac{1}{4}$ of these states $\rho_{i}$ (w.h.p.\ over
the choice of the table $E$) will have the property that \begin{equation}
\Tr[H^{(i)}\rho_{i}]\geq\frac{1}{4\sqrt{m}}+O\left(\frac{1}{m^{2}}\right)\label{quarter}\end{equation}
 and the rest have \begin{equation}
p_{1}(\rho_{i})\geq\frac{1}{2}-O\left(\frac{1}{m}\right)\label{3quarters}\end{equation}
 which implies that \[
\E_{i}p_{1}(\rho_{i})\geq\frac{1}{2}+\frac{1}{8\sqrt{m}}+O\left(\frac{1}{m^{2}}\right).\]

We use the state \[
\rho=\rho_{1}\otimes\rho_{2}\otimes...\otimes\rho_{l}\]
 as our forged quantum money. If the verifier selects $j\left(i\right)$
at random and measures $Q$ (from equation \ref{eq:operatorQ_appendix}),
then the expected outcome is at least $\frac{1}{4}(\frac{1}{4\sqrt{m}}+O(\frac{1}{m^{2}}))+\frac{3}{4}O(\frac{1}{m}$),
and the probability of an outcome less than $\frac{1}{32\sqrt{m}}$
(for $\epsilon\le\frac{1}{16\sqrt{m}}$, the verifier can only reject
if this occurs) is exponentially small for $m$ sufficiently large
by independence of the registers. Therefore the forged money state
$\rho$ is accepted by Aaronson's verifier with probability that is
exponentially close to 1 if $\epsilon\leq\frac{1}{16\sqrt{m}}$.

Before describing our algorithm to generate the states $\left\{ \rho_{i}\right\} $,
we must understand the statistics (in particular, we consider the
first two moments) of each $H^{\left(i\right)}$ on the fully mixed
state $\frac{I}{2^{n}}$. We will assume that, for $j\ne k$, $E_{ij}\ne E_{ik}$.
We also assume that the operators $\pm I\otimes I\otimes I...\otimes I$
do not appear in the list. Both of these assumptions are satisfied
with overwhelming probability. The first and second moments of $H^{(i)}$
are\[
\Tr\left[H^{(i)}\frac{I}{2^{n}}\right]=0\]
 and \begin{align}
 & \quad\Tr\left[\left(H^{\left(i\right)}\right)^{2}\frac{I}{2^{n}}\right]\\
 & =2^{-n}\Tr\left[\frac{1}{m^{2}}\sum_{j}\left(E_{i,j}\right)^{2}+\frac{1}{m^{2}}\sum_{j\ne k}E_{i,j}E_{i,k}\right]\nonumber \\
 & =\frac{1}{m}.\label{eq:2ndmoment}\end{align}
 Now let us define $f_{i}$ to be the fraction (out of $2^{n}$) of
the eigenstates of $H^{(i)}$ which have eigenvalues in the set $[\frac{1}{2\sqrt{m}},1]\cup[-1,-\frac{1}{2\sqrt{m}}]$.
Since the eigenvalues of $H^{(i)}$ are bounded between $-1$ and
$1$, we have\[
\Tr\left[\left(H^{\left(i\right)}\right)^{2}\frac{I}{2^{n}}\right]\leq f_{i}+(1-f_{i})\frac{1}{4m}.\]
 Plugging in equation \ref{eq:2ndmoment} and rearranging we obtain
\[
f_{i}\geq\frac{3}{4m-1}.\]
 We also define $g_{i}$ to be the fraction of eigenstates of $H^{(i)}$
that have eigenvalues in the set $[\frac{1}{2\sqrt{m}},1]$. The distribution
(for any fixed $i$) of $E_{ij}$ as generated by the bank is symmetric
under negation of all the $E_{ij}$, so with probability at least
$\nicefrac{1}{2}$ over the choice of the operators in the row labeled
by $i$, the fraction $g_{i}$ satisfies\begin{equation}
g_{i}\geq\frac{3}{8m-2}.\label{eq:g_i}\end{equation}
 We assume this last inequality is satisfied for at least $\nicefrac{1}{4}$
of the indices $i\in[l]$, for the particular table $E_{ij}$ that
we are given. The probability that this is not the case is exponentially
small in $l$.

Ideally, we would generate the states $\rho_{i}$ by preparing the
fully mixed state, measuring $H^{\left(i\right)},$ keeping the result
if the eigenvalue is at least $\frac{1}{2\sqrt{m}},$ and otherwise
trying again, up to some appropriate maximum number of tries. After
enough failures, we would simply return the fully mixed state. It
is easy to see that outputs of this algorithm would satisfy eq.~\ref{eq:bruteforce-low-epsilon-cond}
with high probability.

Unfortunately, we cannot efficiently measure the exact eigenvalue
of an arbitrary Hermitian operator, but we can use phase estimation,
which gives polynomial error using polynomial resources. In appendix~\ref{sub:Review-of-the}
we review the phase estimation algorithm which is central to our procedure
for generating the states $\rho_{i}$. In section~\ref{sub:Preparing-the-states},
we describe an efficient algorithm to generate $\rho_{i}$ using phase
estimation and show that the resulting states, even in the presence
of errors due to polynomial-time phase estimation, are accepted by
the verifier with high probability, assuming that the table $E_{ij}$
has the appropriate properties.

\subsection{Procedure to Generate $\rho_{i}$\label{sub:Preparing-the-states}}

We now fix a particular value of $i$ and, for convenience, define
$H=\frac{1}{4}H^{(i)}$ so that all the eigenvalues of $H$ lie in
the interval $[-\frac{1}{4},\frac{1}{4}]$. We denote the eigenvectors
of $H$ by $\{|\psi_{j}\rangle\}$ and write \[
e^{2\pi iH}|\psi_{j}\rangle=e^{2\pi i\phi_{j}}|\psi_{j}\rangle.\]
The positive eigenvalues of $H$ map to phases $\phi_{j}$ in the
range$[0,\frac{1}{4}]$ and negative eigenvalues of $H$ map to $[\frac{3}{4},1]$.

We label each eigenstate of $H$ as either {}``good'' or {}``bad''
according to its energy. We say an eigenstate $|\psi_{j}\rangle$
is good if $\phi_{j}\in[\frac{1}{16\sqrt{m}},\frac{1}{4}].$ Otherwise
we say it is bad (which corresponds to the case where $\phi_{j}\in[0,\frac{1}{16\sqrt{m}})\cup[\frac{3}{4},1]).$

We use the following algorithm to produce a mixed state $\rho_{i}$.
\begin{enumerate}
\item Set $k=1$. 
\item Prepare the completely mixed state $\frac{I}{2^{n}}$. In our analysis
of this step, we will imagine that we have selected an eigenstate
$|\psi_{p}\rangle$ of $H$ uniformly at random, which yields identical
statistics.
\item Use the phase estimation circuit to measure the phase of the operator
$e^{2\pi iH}$. Here the phase estimation circuit (see appendix~\ref{sub:Review-of-the})
acts on the original $n$ qubits in addition to $q=r+\lceil\log(2+\frac{2}{\delta})\rceil$
ancilla qubits, where we choose \begin{align*}
r & =\lceil\log(20m)\rceil\\
\delta & =\frac{1}{m^{3}}.\end{align*}

\item Accept the resulting state (of the $n$ qubit register) if the measured
phase $\phi^{\prime}=\frac{z}{2^{q}}$ is in the interval $[\frac{1}{8\sqrt{m}}-\frac{1}{20m},\frac{1}{2}].$
In this case stop and output the state of the first register. Otherwise
set $k=k+1$. 
\item If $k=m^{2}+1$ then stop and output the fully mixed state. Otherwise
go to step 2. 
\end{enumerate}
We have chosen the constants in steps 3 and 4 to obtain an upper bound
on the probability $p_{b}$ of accepting a bad state in a particular
iteration of steps 2, 3, and 4:\begin{align*}
p_{b} & =\Pr\left(|\psi_{p}\rangle\text{ is bad and you accept }\right)\\
 & \leq\Pr\left(\text{accept given that }|\psi_{p}\rangle\text{ was bad}\right)\\
 & \leq\Pr\left(|\phi_{p}-\phi^{\prime}|>\frac{1}{16\sqrt{m}}-\frac{1}{20m}\right)\\
 & \leq\Pr\left(|\phi_{p}-\phi^{\prime}|>\frac{1}{20m}\right)\\
 & \leq\delta\mbox{ by equation \ref{eq:errorbound}.}\end{align*}

Above, we considered two cases depending on whether or not the inequality
\ref{eq:g_i} is satisfied for the register $i$. We analyze the algorithm
in these two cases separately.

\subsubsection*{Case 1: Register $i$ satisfies inequality~\ref{eq:g_i}}

In this case, choosing $p$ uniformly, \begin{equation}
\Pr\left(\frac{1}{4}\geq\phi_{p}\geq\frac{1}{8\sqrt{m}}\right)\geq\frac{3}{8m-2}\label{eq:pr_ineq}\end{equation}
This case occurs for at least $\nicefrac{1}{4}$ of the indices $i\in[l]$
with all but exponential probability.

The probability $p_{g}$ that you pick a good state (in a particular
iteration of steps 2, 3, and 4) and then accept it is at least\begin{align*}
p_{g} & =\Pr\left(|\psi_{p}\rangle\text{ is good and you accept}\right)\\
 & \geq\Pr\left(\frac{1}{4}\geq\phi_{p}\geq\frac{1}{8\sqrt{m}}\text{ and you accept}\right)\\
 & =\Pr\left(\frac{1}{4}\geq\phi_{p}\geq\frac{1}{8\sqrt{m}}\right)\\
 & \qquad\times\Pr\left(\text{accept given }\frac{1}{4}\geq\phi_{p}\geq\frac{1}{8\sqrt{m}}\right)\\
 & \geq\Pr\left(\frac{1}{4}\geq\phi_{p}\geq\frac{1}{8\sqrt{m}}\right)(1-\delta)\\
 & \geq\frac{3}{8m-2}\left(1-\frac{1}{m^{3}}\right)\\
 & \geq\frac{1}{4m}\text{ , for m sufficiently large.}\end{align*}
Thus the total probability of outputting a good state is (in a complete
run of the algorithm)\begin{align}
 & \quad\Pr(\text{output a good state})\\
 & =\sum_{k=1}^{m^{2}}p_{g}(1-p_{g}-p_{b})^{k-1}\nonumber \\
 & =\frac{p_{g}}{p_{g}+p_{b}}\left(1-(1-p_{g}-p_{b})^{m^{2}}\right)\nonumber \\
 & \geq\frac{p_{g}}{p_{g}+p_{b}}\left(1-(1-p_{g})^{m^{2}}\right)\nonumber \\
 & \geq\frac{p_{g}}{p_{g}+\delta}\left(1-(1-p_{g})^{m^{2}}\right).\nonumber \\
 & \geq\frac{p_{g}}{p_{g}+\delta}\left(1-e^{-p_{g}m^{2}}\right)\label{eq:outputgood_prob-1}\\
 & \geq\frac{1}{1+\frac{4}{m^{2}}}\left(1-e^{-p_{g}m^{2}}\right)\text{ for m sufficiently large.}\nonumber \\
 & =1-O\left(\frac{1}{m^{2}}\right)\nonumber \end{align}
 So in this case, the state $\rho_{i}$ will satisfy\begin{align*}
 & \quad\Tr\left[H^{(i)}\rho_{i}\right]\\
 & \geq\Pr\left(\text{output a good state}\right)\frac{1}{4\sqrt{m}}\\
 & \quad-\left(1-Pr\left(\text{output a good state}\right)\right)\\
 & =\frac{1}{4\sqrt{m}}+O\left(\frac{1}{m^{2}}\right).\end{align*}

\subsubsection*{Case 2: Register $i$ does not satisfy inequality~\ref{eq:g_i}}

This case occurs for at most $\nicefrac{3}{4}$ of the indices $i\in[l]$
with all but exponentially small probability.

The probability of accepting a bad state for register $i$ at any
point is \begin{equation}
\Pr\left(\text{accept a bad state ever}\right)\le\sum_{k=1}^{m^{2}}\delta=\frac{1}{m}.\label{eq:failurebound}\end{equation}
 So the state $\rho_{i}$ which is generated by the above procedure
will satisfy\begin{align*}
 & \quad\Tr\left[H^{(i)}\rho_{i}\right]\\
 & \geq-\Pr\left(\text{accept a bad state ever}\right)\\
 & =-\frac{1}{m}.\end{align*}

We have thus shown that equation \ref{quarter} holds for all indices
$i$ which satisfy inequality \ref{eq:g_i} and that equation \ref{3quarters}
holds for the rest of the indices. As discussed above, this guarantees
(assuming at least $\nicefrac{1}{4}$ of the indices $i$ satisfy
inequality \ref{eq:g_i}) that our forged state $\rho=\rho_{1}\otimes\rho_{2}\otimes...\otimes\rho_{l}$
is accepted by the verifier with high probability if $\epsilon\leq\frac{1}{16\sqrt{m}}$.

\subsection{Review of the Phase Estimation Algorithm \label{sub:Review-of-the} }

In this section we review some properties of the phase estimation
algorithm as described in \citet{nielsen-chuang}. We use this algorithm
in appendix~\ref{sec:details-of-low-epsilon-attack} to measure the
eigenvalues of the operator $e^{2\pi iH}$. The phase estimation circuit
takes as input an integer $r$ and a parameter $\delta$ and uses
\[
q=r+\lceil log(2+\frac{2}{\delta})\rceil\]
 ancilla qubits. When used to measure the operator $e^{2\pi iH}$,
phase estimation requires as a subroutine a circuit which implements
the unitary operator $e^{2\pi iHt}$ for $t\leq2^{r}$, which can
be approximated efficiently if $2^{r}=poly(n)$. This approximation
of the Hamiltonian time evolution incurs an error which can be made
polynomially small in $n$ using polynomial resources (see for example
\citet{nielsen-chuang}). We therefore neglect this error in the remainder
of the discussion. The phase estimation circuit, when applied to an
eigenstate $|\psi_{j}\rangle$ of $H$ such that \[
e^{2\pi iH}|\psi_{j}\rangle=e^{2\pi i\phi_{j}}|\psi_{j}\rangle,\]
 and with the $q$ ancillas initialized in the state $|0\rangle^{\otimes q}$,
outputs a state \[
|\psi_{j}\rangle\otimes|a_{j}\rangle\]
 where $|a_{j}\rangle$ is a state of the ancillas. If this ancilla
register is then measured in the computational basis, the resulting
$q$ bit string $z$ will be an approximation to $\phi_{j}$ which
is accurate to $r$ bits with probability at least $1-\delta$ in
the sense that \begin{equation}
\Pr\left(\left|\phi_{j}-\frac{z}{2^{q}}\right|>\frac{1}{2^{r}}\right)\leq\delta.\label{eq:errorbound}\end{equation}
 In order for this algorithm to be efficient, we choose $r$ and $\delta$
so that $2^{r}=\poly(n)$ and $\delta=\frac{1}{\poly(n)}$.

\section{Insecurity of the Stabilizer Money for $\epsilon\geq\frac{c}{\sqrt{m}}$
\label{sec:details-for-large-epsilon}}

In this section, we will describe how to forge the Stabilizer Money
when the number of commuting measurements is at least $c\sqrt{m}$
for any constant $c>0$. We will consider each column of the table
separately. For the $i^{\text{th}}$ column, let $M=M_{i}$ be the
list of possible measurements for $\psi=\psi_{i}$, and let $K=K_{i}$
denote the set of commuting measurements that stabilize $\psi$. Set
$k=|K|$ and $m=|M|$. We will first consider the case $k>100\sqrt{m}$,
and we will then show how to reduce the case $k>c\sqrt{m}$ to this
case for any constant $c>0$. The algorithm we present has success
probability $4/5$ over the choice of the random measurements. We
have not attempted to optimize this probability, and it could be improved
with a more careful analysis.

We begin by casting our question as a graph problem. Let $G$ be a
graph whose vertices correspond to the $m$ measurements, and connect
vertices $i$ and $j$ if and only if the corresponding measurements
commute. The set $K$ now forms a clique, and we aim to find it.

In general, it is intractable to find the largest clique in a graph.
In fact, it is NP-hard even to approximate the size of the largest
clique within $n^{1-\epsilon}$, for any $\epsilon>0$~\citet{Zuckerman}.
However, if the graph is obtained by planting a clique of size $\epsilon\sqrt{m}$
in an (Erdös-Rényi) random graph drawn from $G(m,1/2)$, Alon, Krivelevich,
and Sudakov showed that one can find the clique in polynomial time
with high probability~\citet{AKS98}. Unfortunately, the measurement
graph $G$ is not drawn from $G(m,1/2)$, so we cannot directly apply
their result. However, we shall show that $G$ is sufficiently random
that a modified version of their approach can be made to go through.
The main tool that we use is to show that $G$ is $k$-wise independent
and that this is enough for a variant of the clique finding algorithm
to work. $k$ wise independent random graphs were studied by \citet{AN08},
although they were interested in other properties of them.

\subsection{Properties of the Measurement Graph}

\global\long\def\ip#1#2{\langle#1, #2\rangle}
 \global\long\def\0{\mathbf{0}}
 \global\long\def\F{\mathbb{F}}

To analyze $G$, it will be convenient to use a linear algebraic description
of its vertices and edges. Recall that any stabilizer measurement
on $n$ qubits can be described as a vector in $\mathbb{F}_{2}^{2n}$
as follows: 
\begin{itemize}
\item for $j\leq n$, set the $j^{\text{th}}$ coordinate to 1 if and only
if the operator restricted to the $j^{\text{th}}$ qubit is $X$ or
$Y$, and 
\item for $n<j\leq2n$, set the $j^{\text{th}}$ coordinate to 1 if and
only if the operator restricted to the $\left(j-n\right)^{\text{th}}$
qubit is $Y$ or $Z$.
\end{itemize}
For $v,w\in\F_{2}^{2n}$, let \[
\ip vw=v^{T}\left(\begin{array}{cc}
\0_{n} & I_{n}\\
I_{n} & \0_{n}\end{array}\right)w,\]
 where $I_{n}$ and $\0_{n}$ are the $n\times n$ identity and all-zeros
matrices, respectively. It is easy to check that the stabilizer measurements
corresponding to $v$ and $w$ commute if and only if $\ip vw=0$
(over $\F_{2}$).

Using this equivalence between Pauli group operators and vectors,
each vertex $u$ of the graph $G$ is associated with a vector $s_{u}$.
There is an edge between vertices $u$ and $v$ in $G$ if and only
if $\ip{s_{u}}{s_{v}}=0$. This means that the $2mn$ bits that encode
the vectors $\left\{ s_{u}\right\} $ also encode the entire adjacency
matrix of $G$. There are $m\left(m-1\right)/2$ possible edges in
$G$, so the distribution of edges in $G$ is dependent (generically,
$m\left(m-1\right)/2)>2mn$). Fortunately, this dependence is limited,
as we can see from the following lemma.

\begin{lemma}\label{ind} Let $v_{1},\ldots v_{t},u$ be measurements
such that $s_{v_{1}},\ldots s_{v_{t}},s_{u}$ are linearly independent,
and let $x_{1},\ldots,x_{t}\in\{0,1\}$ be arbitrary. Let $v$ be
a random stabilizer measurement such that $\ip{s_{v}}{s_{v_{i}}}=x_{i}$
for every $i$ and the vectors $s_{v_{1}},\dots,s_{v_{t}},s_{u},s_{v}$
are linearly independent. Then \[
\Pr(\ip{s_{v}}{s_{u}}=0)=1/2\pm O\left(\frac{1}{2^{2(n-t)}}\right).\]
 \end{lemma}

\begin{proof}

The vector $s_{v}\in\{0,1\}^{2n}$ is chosen uniformly at random from
the set of vectors satisfying the following constraints: 
\begin{enumerate}
\item For every $i$, we have $\ip{s_{v}}{s_{v_{i}}}=x_{i}$. 
 
\item The vectors $s_{v_{1}},\ldots s_{v_{t}},s_{u},s_{v}$ are linearly
independent. 

\end{enumerate}
Let $S_{0}$ denote the set of vectors that satisfy these constraints
and have $\ip{s_{v}}{s_{u}}=0$, and let $S_{1}$ be the set of vectors
that satisfy these constraints and have $\ip{s_{v}}{s_{u}}=1$. We
have \[
\Pr(\ip{s_{v}}{s_{u}}=0)=\frac{|S_{0}|}{|S_{0}+S_{1}|}.\]
The vectors $s_{v_{1}},\ldots s_{v_{t}},s_{u}$ are linearly independent,
so there are $2^{2n-t-1}$ solutions to the set of equations $\ip{s_{v}}{s_{u}}=1$
and $\ip{s_{v}}{s_{v_{i}}}=x_{i}$ for all $i$. This implies that
$|S_{1}|\le2^{2n-t-1}$.

Constraint 2 rules out precisely the set of vectors in the span of
$s_{v_{1}},\dots,s_{v_{t}},s_{u}$. This is a $(t+1)$-dimensional
subspace, so it contains $2^{t+1}$ points, and thus 
 $|S_{0}|\ge2^{2n-t-1}-2^{t+1}$. It follows that \begin{align*}
\Pr(\ip{s_{v}}{s_{u}}=0) & \ge\frac{2^{2n-t-1}-2^{t+1}}{2^{2n-t}-2^{t+1}}\\
 & =\frac{1}{2}-\frac{1}{2^{2n-2t}-1}\\
 & =\frac{1}{2}-O\left(\frac{1}{2^{2(n-t)}}\right).\end{align*}
 Repeating this argument gives the same bound for $\Pr(\ip{s_{v}}{s_{u}}=1)$,
from which the desired result follows. \end{proof}


\subsection{Finding Planted Cliques in Random Graphs}

\label{AKSsection} Our algorithm for finding the clique $K$ will
be identical to that of Alon, Krivelevich, and Sudakov~\citet{AKS98},
but we will need to modify the proof of correctness to show that it
still works in our setting. In this section, we shall give a high
level description of~\citet{AKS98} and explain the modifications
necessary to apply it to $G$. The fundamental difference is that
Alon et al.\ rely on results from random matrix theory that use the
complete independence of the matrix entries to bound mixed moments
of arbitrarily high degree, but we only have guarantees about moments
of degree $O(\log m)$. As such, we must adapt the proof to use only
these lower order moments.


Let $G(m,1/2,k)$ be a random graph from $G(m,1/2)$ augmented with
a planted clique of size $k$, and let $A$ be its adjacency matrix.
Let $\lambda_{1}\ge\lambda_{2}\ge\dots\ge\lambda_{m}$ be the eigenvalues
of $A$, and let $v_{1},\dots,v_{m}$ be the corresponding eigenvectors.
To find the clique, Alon et al.\ find the set $W$ of vertices with
the $k$ largest coordinates in $v_{2}$. They then prove that, with
high probability, the set of vertices that have at least $3k/4$ neighbors
in $W$ precisely comprise the planted clique.

The analysis of their algorithm proceeds by analyzing the largest
eigenvalues of $A$. They begin by proving that the following two
bounds hold with high probability: 
\begin{itemize}
\item $\lambda_{1}\geq\left(\frac{1}{2}+o(1)\right)m,$ and 
\item $\lambda_{i}\leq\left(1+o(1)\right)\sqrt{m}$ for all $i\geq3$. 
\end{itemize}
The second of these bounds relies heavily on a result by Füredi and
Komlós about the eigenvalues of matrices with independent entries.
The independence assumption will not apply in our setting, and thus
we will need to reprove this bound for our graph $G$. This is the
main modification that we will require to the analysis of~\citet{AKS98}.

They then introduce a vector $z$ that has $z_{i}=(m-k)$ when vertex
$i$ belongs to the planted clique, and has $z_{i}=-k$ otherwise.
Using the above bounds, they prove that, when one expands $z$ in
the eigenbasis of $A$, the coefficients of $v_{1},v_{3},\dots,v_{m}$
are all small compared to $||z||$, so $z$ has most of its norm coming
from its projection onto $v_{2}$. This means that $v_{2}$ has most
of its weight on the planted clique, which enables them to prove the
correctness of their algorithm.

Other than the bound on $\lambda_{3},\dots,\lambda_{m}$, the proof
goes through with only minor changes. The bound on $\lambda_{1}=(1+o(1))m/2$,
follows from a simple analysis of the average degree, which holds
for the measurement graph as well. The rest of their proof does not
make heavy use of the structure of the graph. The only change necessary
is to replace various tail bounds on the binomial distribution and
Chebyschev bounds with Markov bounds. These weaker bounds result in
a constant failure probability and weaker constants, but they otherwise
do not affect the proof. (For brevity, we omit the details.) As such,
our remaining task is to bound $\lambda_{i}$ for $i\geq3$.

\subsection{Bounding $\lambda_{3},\dots,\lambda_{m}$}

\label{rand} To bound the higher eigenvalues of the adjacency matrix,
Alon et al.\ apply the following theorem of Füredi and Komlós~\citet{FK81}:
\begin{lemma} Let $R$ be a random symmetric $m\times m$ matrix
in which $R_{i,i}=0$ for all $i$, and the other entries are independently
set to $\pm1$ with $\Pr(R_{i,j}=1)=\Pr(R_{i,j}=-1)=\frac{1}{2}$.
The largest eigenvalue of $R$ is at most $m+O(m^{1/3}\log m)$ with
high probability. \end{lemma}

We will prove a slightly weaker variant of this lemma for random measurement
graphs. Let $B$ be a matrix that is generated by picking $m$ random
stabilizer measurements $M_{1},\ldots,M_{m}$ and setting $B_{i,i}=0$,
$B_{i,j}=1$ if $M_{i}$ commutes with $M_{j}$, and $B_{i,j}=-1$
if $M_{i}$ anticommutes with $M_{j}$. The main technical result
of this section will be the following: \begin{thm} \label{largesteigenval}
With high probability, the largest eigenvalue of $B$ is at most $10\sqrt{m}.$
\end{thm}

Alon et al.\citet{AKS98} show how to transform a bound on the eigenvalues
of $R$ into a bound on the third largest eigenvalue of $A$. This
reduction does not depend on the properties of $G$, and it works
in our case when applied to $B$. This gives a bound of $10\sqrt{m}$
on the third largest eigenvalue of the adjacency matrix of $G$.

The proof of Theorem~\ref{largesteigenval} will rely on the following
lemma, which shows that the entries of small powers of the matrix
$B$ have expectations quite close to those of $R$. \begin{lemma}\label{lem:expect}
For $t\le O(\log m)$, \[
\mathbb{E}\left[(B^{t})_{i,j}\right]=\mathbb{E}\left[(R^{t})_{i,j}\right]\pm\frac{1}{2^{\Omega(n-t)}}.\]
 \end{lemma}

\begin{proof}{[}Proof of Lemma~\ref{lem:expect}{]} With high probability,
for every subset of vertices $U$ such that $|U|<t\leq O(\log m)$,
we have that the set $\{s_{u}\,|\, u\in U\}$ is linearly independent
over $\F_{2}$. We condition the rest of our analysis on this high
probability event.

We begin by expanding the quantity we aim to bound: \begin{align}
\mathbb{E}\left[(B^{t})_{i,j}\right] & =\mathbb{E}\left[\sum_{\ell_{2},\ldots\ell_{t}}\prod_{\alpha=1}^{t+1}B_{\ell_{\alpha},\ell_{\alpha+1}}\right]\nonumber \\
 & =\sum_{\ell_{2},\ldots\ell_{t}}\mathbb{E}\left[\prod_{\alpha=1}^{t+1}B_{\ell_{\alpha},\ell_{\alpha+1}}\right]\label{expansion}\end{align}
 where we take set $\ell_{1}=i$ and $\ell_{t+1}=j$, and we sum over
all possible values of the indices $\ell_{2},\dots,\ell_{t}$.

We break the nonzero terms in this summation into two types of monomials:
those in which every matrix element appears an even number of times,
and those in which at least one element appears an odd number of times.
In the former case, the monomial is the square of a $\pm1$-valued
random variable, so we have \[
\mathbb{E}\left[\prod_{\alpha}B_{\ell_{\alpha},\ell_{\alpha+1}}\right]=\mathbb{E}\left[\prod_{\alpha}R_{\ell_{\alpha},\ell_{\alpha+1}}\right]=1,\]
 and it suffices to focus on the latter case. By the same reasoning,
we can drop any even number of occurrences of an element, so it suffices
to estimate the expectations of monomials of degree at most $t$ in
which all of the variables are distinct.

Any such monomial in the $R_{i,j}$ has expectation zero by symmetry,
so we need to provide an upper bound on terms of the form $\prod_{\alpha=1}^{q}B_{\ell_{\alpha},\ell_{\alpha+1}}$,
where $q\le t\le r$ and each matrix element appears at most once.

Consider the probability that $B_{q-1,q}=1$, where we take the probability
over the choice of the $2n$ bit string $s_{q}$, given that for any
$\alpha\le q$, we have $B_{\alpha,\alpha+1}=x_{\alpha}$ for some
value $x_{\alpha}$. We are computing this expectation conditioned
on the the $s_{u}$ being linearly independent, so we can apply Lemma~\ref{ind}.
This gives\begin{align*}
 & \mathbb{E}\prod_{\alpha=1}^{q}B_{\ell_{\alpha},\ell_{\alpha+1}}\\
= & \sum_{x_{1},\ldots x_{q-1}}\Pr(\ip{s_{\ell_{\alpha}}}{s_{\ell_{\alpha+1}}}=x_{\alpha})\\
 & \quad\times\Big\{\Pr(\ip{s_{q-1}}{s_{q}}=1|x_{1},\ldots x_{q-1})\\
 & \quad-\Pr(\ip{s_{q-1}}{s_{q}}=-1|x_{1},\ldots x_{q-1})\Big\}\\
\le & O\left(\frac{1}{2^{2(n-t)}}\right)\cdot\sum_{x_{1},\ldots x_{q-1}}\Pr(\ip{s_{\ell_{\alpha}}}{s_{\ell_{\alpha+1}}}=x_{\alpha})\\
= & O\left(\frac{1}{2^{2(n-t)}}\right).\end{align*}
 There are $n^{O(\log m)}$ terms in the summation of eq.~, and we
have shown that each term is at most $O\left(1/{2^{2(n-t)}}\right)$,
so we obtain \[
\mathbb{E}\left[(B^{t})_{i,j}\right]\leq O\left(\frac{n^{O(\log m)}}{2^{2(n-t)}}\right)=\frac{1}{2^{\Omega(n)}},\]
 as desired. \end{proof}

We can now use this lemma to prove Theorem~\ref{largesteigenval}.
\begin{proof}{[}Proof of Theorem~\ref{largesteigenval}{]} Consider
a random matrix $R$, with $R_{i,i}=0$ and each other cell distributed
independently at random according to $\Pr(R_{i,j}=1)=\Pr(R_{i,j}=-1)=\frac{1}{2}$.
Lemma $3.2$ of~\citet{FK81} shows that, for $t<m^{1/3}$, \[
\Tr(\mathbb{E}(R^{t}))=m^{t/2+1}4^{t}.\]
 For $t\ge10\log m$, Lemma \ref{lem:expect} implies that \begin{align*}
\Tr(\mathbb{E}(B^{t}))=\Tr(\mathbb{E}(R^{t}))\pm\frac{1}{2^{\Omega(n-t)}}\\
=m^{t/2+1}4^{t}\pm\frac{1}{2^{\Omega(n-t)}}.\end{align*}
 Let $\lambda_{1}\ge\dots\ge\lambda_{n}$ be the eigenvalues of $B$.
For any even $t$, one has that \[
\Tr B^{t}=\sum_{i}\lambda_{i}^{t}\ge\lambda_{1}^{t}.\]
 Applying this relation with $t=10\log m$ gives: \begin{align*}
\Pr(\lambda_{1}\ge10\sqrt{m}) & =\Pr\left(\lambda_{1}^{t}\ge(10\sqrt{m})^{t}\right)\\
\le(10\sqrt{m})^{-t}\mathbb{E}\lambda_{1}^{t} & \le(10\sqrt{m})^{-t}m^{t/2+1}4^{t}\\
 & =m\left(\frac{4}{10}\right)^{t}<1/m^{4}.\qedhere\end{align*}
 \end{proof}

Plugging the bound from Theorem~\ref{largesteigenval} into the argument
from the section~\ref{AKSsection} and computing the correct constants
yields that the algorithm finds a planted clique in $G$ of size at
least $100\sqrt{m}$ with probability $4/5$.

\subsection{Finding Cliques of Size $c\sqrt{m}$}

To break stabilizer money for all $\epsilon\geq\frac{c}{\sqrt{m}}$,
we extend our algorithm to find cliques of size $c\sqrt{m}$ for any
$c>0$. In~\citet{AKS98}, Alon et al.\ show how to bootstrap the
above scheme to work for any $c$.

The procedure used by Alon et al.\ is to iterate over all sets of
vertices of size $\log(100/c)$, and, for each such set $S$, to try
to find a clique in the graph $G_{S}$ of the vertices that are connected
to all of the vertices in $S$.

When $S$ is in the planted clique, $G_{S}$ also contains the clique.
However, $|G_{S}|\approx c|G|/100$, as most of the vertices that
are outside the clique are removed. As $G_{S}$ behaves like a random
graph with the same distribution as the original graph but with a
planted clique of size $100\sqrt{|G_{S}|}$, one can find it using
the second largest eigenvector.

To use the same algorithm in our case, we apply Lemma \ref{lem:expect}
with parameter $k+\log100/c$. This shows that, up to a small additive
error, the expected value of the $k^{\text{th}}$ power of the adjacency
matrix of $G_{S}$ behaves like the expected value of the $k^{\text{th}}$
power of the adjacency matrix of a random graph, which was all that
we used in the proof.
\end{document}